\documentclass[12pt]{amsart}
\usepackage[]{color}
\usepackage{setspace}
\usepackage{wrapfig}
\usepackage{lineno}
\usepackage{graphicx}
\usepackage{amsmath}
\usepackage{amsfonts}
\usepackage{hyperref}
\usepackage{subcaption}
\hypersetup{
    colorlinks=true,
    linkcolor=blue,
    filecolor=magenta,
    urlcolor=cyan,
    citecolor=blue,
}
\usepackage{enumitem}

\usepackage{url}
    \usepackage{amsmath,amsxtra,amssymb,latexsym,epsfig,amscd,amsthm,fancybox,epsfig}
\usepackage[mathscr]{eucal}
\usepackage{graphicx}
\usepackage{multicol,xcolor}
\usepackage{cases}
\usepackage{color}
\usepackage{hyperref}
\usepackage{bigints}

\DeclareFontFamily{U}{matha}{\hyphenchar\font45}
\DeclareFontShape{U}{matha}{m}{n}{
      <5> <6> <7> <8> <9> <10> gen * matha
      <10.95> matha10 <12> <14.4> <17.28> <20.74> <24.88> matha12
      }{}
\DeclareSymbolFont{matha}{U}{matha}{m}{n}

\DeclareMathSymbol{\Lt}{3}{matha}{"CE}
\DeclareMathSymbol{\Gt}{3}{matha}{"CF}

\newtheorem{thm}{Theorem}[section]
\newtheorem {asp}{Assumption}[section]

\newtheorem{rmk}{Remark}[section]

\newtheorem{cor}{Corollary}[section]

\newtheorem{prop}{Proposition}[section]
\theoremstyle{definition}

\theoremstyle{remark}

\DeclareMathOperator{\NN}{\mathbf N}

\DeclareMathOperator{\Conv}{Conv}

\newcommand{\M}{\mathcal{M}}

\newcommand{\E}{\mathbb{E}}

\newcommand{\BX}{\mathbf{X}}
\newcommand{\bx}{\mathbf{x}}

\newcommand{\PP}{\mathbb{P}}

\newcommand{\R}{\mathbb{R}}

\DeclareMathOperator{\Var}{Var}

\newcommand{\bed}{\begin{displaymath}}
\newcommand{\eed}{\end{displaymath}}
\newcommand{\bea}{\bed\begin{array}{rl}}
\newcommand{\eea}{\end{array}\eed}

\newcommand{\barray}{\begin{array}{ll}}
\newcommand{\earray}{\end{array}}
\newcommand{\diag}{{\rm diag}}

\newcommand{\Cov}{\mathrm{Cov}}

\def\bar{\overline}
\def\hat{\widehat}
\def\a.s{\text{\;a.s.\;}}

\newcommand{\Se}{\mathcal{S}}
\newcommand{\Z}{\mathbb{Z}}

\usepackage{natbib}

\begin{document}
\bibliographystyle{agsm}

\title{Population size in stochastic discrete-time ecological dynamics}

\author[A. Hening]{Alexandru Hening }
\address{Department of Mathematics\\
Texas A\&M University\\
Mailstop 3368\\
College Station, TX 77843-3368\\
United States
}
\email{ahening@tamu.edu}

\author[S. Sabharwal]{Siddharth Sabharwal}
\address{Department of Mathematics\\
Texas A\&M University\\
Mailstop 3368\\
College Station, TX 77843-3368\\
United States
}
\email{siddhutifr93@tamu.edu }

\begin{abstract}
We study how environmental stochasticity influences the long-term population size in certain one- and two-species models. The difficulty is that even when one can prove that there is persistence, it is usually impossible to say anything about the invariant probability measure which describes the persistent species. We are able to circumvent this problem for some important ecological models by noticing that the per-capita growth rates at stationarity are zero, something which can sometimes yield information about the invariant probability measure. For more complicated models we use a recent result by Cuello to explore how small noise influences the population size. We are able to show that environmental fluctuations can decrease,  increase, or leave unchanged the expected population size. The results change according to the dynamical model and, within a fixed model, also according to which parameters (growth rate, carrying capacity, etc) are affected by environmental fluctuations. Moreover, we show that not only do things change if we introduce noise differently in a model, but it also matters what one takes as the deterministic `no-noise' baseline for comparison.

\end{abstract}

\keywords{stochastic difference equations, population size, stationarity, coexistence}

\maketitle
%\linenumbers

\tableofcontents

\section{Introduction} \label{s:intro}

In order to have a realistic description of a natural ecosystem one has to include the effects of environmental stochasticity. This is because the environment is constantly changing and this impacts the way various species behave and interact. Although, as a first approximation, it is useful to look at the deterministic dynamics of interacting species, it cannot be expected that true dynamics are captured realistically without introducing random environmental fluctuations \citep{connell1978diversity}. The attractors coming from the deterministic dynamics can be understood as giving us an idea about the averaged behavior of the dynamics while environmental stochasticity can introduce changes that are not only quantitative but can be qualitative. There are examples where the deterministic dynamics predicts coexistence while in the stochastic setting we have one or multiple extinctions. In other examples, the deterministic dynamics predicts extinction, while, by adding stochasticity, one can get a rescue effect where all species coexist \citep{chesson1981environmental, li2016effects}. 

Recently, there have been some important developments for stochastic population dynamics, and there is now a general theory of persistence and extinction for interacting populations \citep{B18, BS19, HNC21, FS24}. An important concept used by this theory is the one of \textit{invasion rate} or \textit{realized per-capita growth rate}, which looks at the growth-rate of an `invader' introduced at low density into a subcommunity of coexisting species which are at stationarity. While this theory tells us, in principle,  exactly which species persist and which go extinct, our focus in this manuscript is different. We want to explore how stochasticity changes the population size of persistent populations. In many situations the expectation is that environmental fluctuations will decrease the population size. For example, \cite{CH02} conjectured that for the single-species Beverton-Holt difference equation a periodic environment is always detrimental for the population: the periodic environment will decrease the average population size. This conjecture was given a rigorous proof in \cite{S06}. A stochastic version of the Cushing-Henson conjecture, in the setting of stochastic difference equations, was proven in \cite{HS05}. One of the main goals of this paper is to look at important single- and multi-species ecological models and explore how environmental stochasticity influences the population size. We will see that it is possible for the population size to \textit{increase} due to environmental fluctuations.

Even though there are general results for the persistence and extinction of stochastic populations, this is usually not enough if one wants to say something about the average population sizes at `stationarity'. For models like the discrete Lotka-Volterra (Ricker) where the logarithm of the growth rate is a linear function of the species densities, one can reduce the problem of finding the average population sizes at stationarity to solving a system of linear equations. In nonlinear models one can sometimes prove whether the population size increases or decreases due to environmental stochasticity but it is usually impossible to have analytic expressions for the change in population size. This happens because the general theory of stochastic coexistence that has recently been developed into a robust mathematical framework \citep{B18, BS19, HNS20, FS24} only gives conditions under which a stationary distribution exists but does not analytically describe this distribution. There are rigorous approximation schemes for the stationary distributions of stochastic differential equations modeling ecological systems \citep{hening2021stationary} but in this paper we want to have analytic expressions that can yield some interesting biological interpretations. In order tackle this problem we look at some recent approximation results by \cite{cuello2019persistence}, where the author showed that under certain natural assumptions, if the noise is small, one can do a type of Taylor approximation around the deterministic stable fixed point (which describes the population sizes in the absence of stochasticity) and find explicitly the first and second order terms. We will use this result to see how some important models behave under the influence of small environmental stochasticity. 

There are related small-noise results which have been used in other settings. \cite{barbier2017cavity} use a disordered systems approach from statistical physics in order to look at community ecology. In particular, they look at small perturbations of the carrying capacities and how these change the abundances of the species. Some of the questions we ask here have already been asked and in part answered \citep{mallmin2024chaotic} for other types of models. These models are usually deterministic and the fluctuations of the population are endogenous, rather than driven by the environment like in our setting. We think that future ecological work should take advantage of both approaches. 

The paper is structured as follows. In Section \ref{s:background} we talk about the models and assumptions we will use, and review the relevant stochastic persistence results. In Section \ref{s:large_noise} we analyze single- and two-species models with arbitrarily large noise, and apply the stochastic persistence results in order to quantify how environmental fluctuations influence the expected population sizes at stationarity. In Section \ref{s:small_noise} we use small-noise expansions of deterministic difference equations to approximate the invariant distribution of persistent systems. This will allow us to quantify how environmental stochasticity changes the population sizes. Numerical experiments and simulations are done in Section \ref{s:sim}. We end our paper with a discussion of our results and future challenges in Section \ref{s:discussion}.

\section{Background} \label{s:background}

\subsection{Models and Assumptions}\label{s:pers}

The basic models we will be working with are stochastic difference equations of the form
\begin{equation}\label{e:SDE}
X_i(t+1) = X_i(t) f_i(\BX(t), \xi(t)).
\end{equation}
Here the vector $\BX(t):=(X_1(t),\dots,X_n(t))\in \Se\subset \R_+^n$ consists of the abundances of the $n$ populations at time $t\in\Z_+$. The random variable $\xi(t)$ quantifies the environmental conditions between time $t$ and $t+1$. The subset $\Se$ is the state space of the dynamics. In most of the examples it will end up being a compact subset of $\R^\ell$ but it can also be all of $\R^\ell$.
The coexistence set is denoted by $\Se_+=\{\bx\in \Se~|~x_i>0, i=1,\dots n\}$ and consists of the subset of the state space where no species is extinct.
The real function $f_i(\BX(t),\xi(t))$ gives the fitness of the $i$-th population at time $t$ and can usually depend on both the population sizes and the environmental conditions. These models have been used extensively in the ecology literature and can include the effects of predation, cannibalism, competition, and seasonal variations.

For simplicity, most of the models considered in this paper will live on a compact state space. However, more general models can also be included (see \cite{HNC21}). We will make the following assumptions throughout the paper:
\begin{itemize}
\item[\textbf{(A1)}] $\xi(0),\dots,\xi(t),\dots$ is a sequence of i.i.d. random variables taking values in $\R^\ell$.
\item[\textbf{(A2)}] For each $i$ the fitness function $f_i(\bx,\xi)$ is continuous in $\bx$ on $\Se$, measurable in $(\bx,\xi)$ and strictly positive.
\item[\textbf{(A3)}] The process $\BX(t)$ returns to compact subsets of $\Se$
exponentially fast, and the growth rates do not change too
abruptly near infinity.
\end{itemize}

Assumptions (A1) and (A2) are needed to ensure that $\BX(t)$ is a Feller process and that $\Se_+$ is an invariant set, i.e. $\BX(t)\in\Se_+, t\in \Z_+$ whenever $\BX(t)\in\Se_+$. Assumption A3) is required (see assumptions in \cite{HNC21}) in order to make sure that $\BX(t)$ does not blow up or fluctuate too abruptly between $0$ and $\infty$. Many ecological models will satisfy these assumptions. 

\begin{rmk}\label{r:compact}
Assumption (A3) is satisfied in particular if, as per the assumptions from \cite{BS19}, we require that
\begin{itemize}
\item [a)] There is a compact subset $K\subset
\R_+^n$ such that all solutions $\BX(t)$
satisfy $\BX(t)\in K$ for $t\in \Z_+$ sufficiently large.
\item [b)] For all $i=1,2,\dots,n $, one has
$\sup_{\bx,\xi}|\ln f_i(\bx,\xi)|<\infty.$
\end{itemize}
These assumptions are natural and are satisfied by many ecological models if the noise is such that the per-capita growth rates are bounded above. For more details, see \cite{HNC21}.
\end{rmk}

\begin{rmk}\label{r:1}
If the dynamics is given by the more general model 
\begin{equation}\label{e:SDE2}
X_i(t+1) = F_i(\BX(t), \xi(t))
\end{equation}
this reduces to \eqref{e:SDE} if $F_i$ is $C^1$ and $F_i(\bx)=0$ whenever $x_i=0$. This is a natural assumption as there is no reason the population should be able to come back from extinction if there is no immigration or dispersal. Under these assumptions we can see that \eqref{e:SDE} is satisfied by setting
\begin{equation*}
   f_i(\bx,\xi) = \begin{cases}
              \frac{F_i(\bx,\xi)}{x_i} & \text{if } x_i>0,\\
               \frac{\partial F_i(\bx,\xi)}{\partial x_i} & \text{if } x_i = 0.
          \end{cases}
\end{equation*}
\end{rmk}

\subsection{Stochastic persistence.}
The extinction set is defined to be the set where at least one species is extinct and therefore is
\[
\Se_0:=\Se\setminus \Se_+=\{\bx\in\Se~:~\min_i x_i=0\}.
\]
The transition operator of the Markov process $\BX(t)$ is given by $P: \mathcal{B}\to \mathcal{B}$ and is the operator which acts on Borel functions $\mathcal{B}:=\{h:\Se\to\R~|~h~\text{Borel}\}$ as
\[
Ph(\bx)=\E_\bx[h(\BX(1))]:=\E[h(\BX(1))~|~\BX(0)=\bx],  ~\bx\in\Se.
\]
The operator $P$ acts by duality on Borel probability measures $\mu$ by $\mu\to \mu P$ where $\mu P$ is the probability measure given by
\[
\int_{\Se}h(\bx) (\mu P)(d\bx):=\int_{\Se}P h(\bx) \mu(d\bx)
\]
for all $h\in C(\Se)$.
A Borel probability measure $\mu$ on $\Se$ is called an \textit{invariant probability measure} if
\[
\mu P = \mu
\]
where $P$ is the transition operator of the Markov process $\BX(t)$. One can think of an invariant probability measure as a way of describing a `random equilibrium': If we start the process with an initial distribution given by the invariant probability measure $\mu$ then the distribution of $\BX(t)$ will still be $\mu$ for all $t\in\Z_+$. 

A key concept is the \textit{realized per-capita growth rate} (also called the \textit{invasion rate}) \citep{SBA11} of species $i$ when introduced at low density in the community described by an invariant probability measure $\mu$
\begin{equation}\label{e:r}
r_i(\mu) =\int_{\R_+^n}\E[\ln f_i(\bx,\xi(1))]\,\mu(d\bx) = \int r_i(\bx)\mu(d\bx)
\end{equation}
where
\[
r_i(\bx) = \E [\ln f_i(\bx,\xi(1))]
\]
is the mean per-capita growth rate of species $i$ at population state $\bx$. This formula says that we need to average the log growth rate over the environment, given here by $\xi(1)$, which we note, is part of the independent identically distributed sequence $\xi(1), \xi(2),\dots$, as well as over the invariant measure $\mu$. The quantity $r_i(\mu)$ tells us whether species $i$ tends to increase or decrease when introduced at an infinitesimally small density into the the subcommunity described by $\mu$. 

For any species $i$ let $\Se^i:=\{\bx,\in\Se~:~x_i>0\}$
be the subset of the state space for which this species has
strictly positive density. If $\mu$ is an ergodic probability measure then
$\mu(\Se^i)\in\{0,1\}$ and one can then define the
\textit{species support} of $\mu$ by $$\Se(\mu)=\{1\leq i\leq
n~:~\mu(\Se^i)=1\}.$$

If the $i$th species is among the ones supported by the subcommunity given by $\mu$, i.e., $i$ lies in the support of $\mu$, then this species is in a sense `at equilibrium' and one can prove  \citep{BS19,HNC21} that its per capita growth rate is zero. The following result gives a sketch of the proof for this fact and also argues that the $r_i(\mu)$ exist for all ergodic/invariant probability measures. 
\begin{prop}\label{p:0}
Let $\mu$ be an invariant probability measure. Then $r_i(\mu)$ exists for all $i$. If $\mu$ is ergodic and $i\in \Se(\mu)$ then
\begin{equation}\label{e:equ}
r_i(\mu)=0.
\end{equation}
\end{prop}
\begin{proof} (Sketch)
The boundedness assumptions from (A3) or from Remark \ref{r:compact} will imply that $\ln f_i(\bx,\xi(1))$ is integrable with respect to $\mu$ and therefore that $r_i(\mu)$ exists.  

Using Assumption (A3) above one can show by making use of the strong law of large numbers for martingales that for $\mu$ almost every $\bx$ we have
 \begin{equation}\label{e5-B3}
 \lim_{t\to\infty}
\frac1t \sum_0^t\left(\log f_i(\BX(t+1))-\E_\bx \log
f_i(\BX(t))\right)=0, \,\text{when}~ \BX(0)=\bx
\end{equation}    
This can be used to show in conjunction with Birkhoff's ergodic theorem that
\[
\lim_{t\to \infty} \frac{\ln X_i(t)}{t} = r_i(\bx)=r_i(\mu), ~\BX(0)=\bx
\]
for $\mu$ almost every $\bx \in \Se^i.$ The ergodicity will then force 
\[
 r_i(\mu) = 0 
\]
as otherwise $X_i(t)\to +\infty$ or $X_i(t)\to 0$.
\end{proof}
This implies that the only species $i$ for which $r_i(\mu)$ can be non-zero are those which are not supported by $\mu$.

The results of \cite{BS19, HNC21} show that the invariant probability measures living on the extinction set $\Se_0$, in many settings, will determine the long term behavior of the system. Loosely speaking, if any such invariant probability measure is a \textit{repeller} which pushes the process away from the boundary in at least one direction, then the species coexist. Let $\Conv(\M)$ denote the set of all invariant probability measures supported on $\Se_0$. In order to have the convergence of the process to a unique stationary distribution one needs some irreducibility conditions which keep the process from being too degenerate \citep{HNC21, MT}. The following theorem is from \cite{HNC21}.
\begin{thm}\label{t:pers1_disc}
Suppose that for all $\mu\in\Conv(\M)$ we have
\begin{equation}\label{e:pers}
\max_{i} r_i(\mu)>0.
\end{equation}
Then the system is stochastically persistent in probability. Under additional irreducibility conditions, there exists a unique invariant probability measure $\pi$ on $\Se_+$ and as $t\to\infty$ the process $\BX(t) $ converges in total variation to $\pi$ whenever $\BX(0)=\bx\in\Se_+$. Furthermore, if $w:\Se_+\to\R$ is bounded
then $$\lim_{t\to\infty}\E w(\BX(t)) = \int_{\Se_+} w(\bx)\,\pi(d\bx).$$
\end{thm}

\begin{rmk}
The irreducibility assumptions make sure there cannot be multiple invariant probability measures on $\Se_+$. This can then be used to prove the convergence of the process to a unique invariant probability measure living on $\Se_+$, analogous to when one can prove the convergence of a unique attractor, usually a fixed point, in the deterministic setting.
For the irreducibility conditions one needs see \cite{MT, B18}. More specifically, in the discrete time setting see \cite{HNC21}. If the random variable $\xi(1)$ is nice (for example has a continuous density etc) then the irreducibility conditions can be shown to hold for many ecological models \citep{BS09, HNC21}.
\end{rmk}

\section{Large noise results}\label{s:large_noise}

\subsection{The Cushing-Henson conjecture} \label{s:CH}
It is well-known that environmental stochasticity can influence the long-term behavior of ecological models. In many settings environmental fluctuations are known to reduce population growth rates \citep{tuljapurkar2013population, schreiber2023partitioning}. It is also important to look at what happens with the population size at stationarity. Will it increase or decrease due to environmental fluctuations? In simple models it seems that environmental fluctuations, both periodic and stochastic, will decrease the population size at stationarity \citep{CH02, HS05, streipert2022derivation}. We explore this phenomenon in various population models. We will show that, because of environmental stochasticity, expected population size at stationarity can increase, decrease , or be unchanged. These results are not only model dependent, but also depend on how environmental stochasticity influences the given model. 

\subsubsection{Beverton-Holt} \label{s:B_H}

One of the most widely used single-species functional responses is the given in the Beverton-Holt model.  This model has been used extensively and can be seen as a discrete time analogue of the logistic equation. Beverton and Holt used it in 1957 to analyze exploited fish populations \citep{beverton2012dynamics} but this functional response has been shown to work well in general with contest and scramble competition models \citep{brannstrom2005role}. 

The Beverton-Holt dynamics is of the form

\[
X_{t+1}= \frac{m K_t}{K_t +(m-1)X_t} X_t.
\]
We assume $K_t$, the carrying capacity, is random and forms an iid sequence $(K_t)_{t\geq 1}$.  Let $\bar K = \E K_1$ and assume that $m>1$. The deterministic system 

\[
\bar X_{t+1}= \frac{m \bar K}{\bar K +(m-1) \bar X_t} \bar X_t
\]
converges as $t\to \infty$ to the fixed point given by the carrying capacity $\bar X_t\to \bar K$. 

If $r(\delta_0)=\ln m >0$, which is the same assumption we needed for the deterministic model above i.e. $m>1$, using Theorem \ref{t:pers1_disc} or the results from \cite{BS09}), one can show that the process $(X_t)$ has a unique stationary distribution $\mu$ on $\R_{++}$. Assume $X_\infty$ is a random variable that has the distribution given by $\mu$. Then by Theorem \ref{t:pers1_disc} $X_t\to X_\infty$ in distribution and, using Proposition \ref{p:0},
\[
r(\mu)=0= \E \ln \left[ \frac{m K_1 }{K_1 +(m-1)X_\infty}\right]\ =\int_{\R_+^n}\E\left[\ln \frac{m K_1 }{K_1 +(m-1)x}\right]\mu(dx).
\]
This yields 
\[
\E \ln (K_1+(m-1)X_\infty) =  \E \ln (m K_1)
\]
Using the fact that the logarithm is strictly concave, the assumption that the random variables are not degenerate, and Jensen's inequality, implies that
\[
 \E \ln (m K_1) < \ln(\E K_1 +(m-1)\E X_\infty).
\]
Taking exponentials and reordering we get
\begin{equation}\label{e:lb} 
  \frac{e^{\E\ln(m K_1)}-\E K_1}{m-1}<\E X_\infty.
\end{equation}
This gives a lower bound on the total expectation size at stationarity. Note that by the proof of the Cushing-Henson conjecture \citep{HS05} the noise is always detrimental and decreases the expected population size in the Beverton-Holt setting, that is 
\begin{equation}\label{e:ub} 
\E X_\infty< \E K_1=\bar K.
\end{equation}
\textbf{Biological interpretation:} As long as $m e^{\E \ln K_1}>\E K_1$, we get that the environmental fluctuations, even though they always decrease the population size at stationarity according to \eqref{e:ub}, cannot decrease the population size too much as \eqref{e:lb} gives us a lower bound. The condition $m e^{\E \ln K_1}>\E K_1$ is equivalent to $\ln m >  \ln \E K_1- \E\ln K_1$. We can expand on this by using some ideas from \cite{aldaz2008selfimprovemvent}. Suppose $Z_1,\dots,Z_n,\dots$ are i.i.d (positive) random variables having the same distribution as $K_1$. Then, using the arithmetic-geometric mean inequality for the square roots of the random variables we get
\[
\left(\sqrt{Z_1}\dots \sqrt{Z_n}\right)^{1/n} \leq \frac{1}{n} (\sqrt{Z_1}+\dots+\sqrt{Z_n})
\]
Taking logs and exponentiating inside the nth root we get
\[
\exp\left(\frac{\ln\sqrt{Z_1}+\dots +\ln \sqrt{Z_n}}{n}\right)\leq \frac{1}{n} (\sqrt{Z_1}+\dots+\sqrt{Z_n}).
\]
Next, use the strong law of large numbers on both sides of the inequality to get
\[
\exp(\E(\ln \sqrt {K_1})) \leq \E(\sqrt{K_1}).
\]
After some elementary manipulations this reduces to the improved estimate
\begin{equation}\label{e:log_Jen_improved}
    \E(\ln K_1) \leq \ln(\E K_1) - \ln\left(\frac{\E K_1}{(\E \sqrt{K_1})^2}\right)
\end{equation}
As a result, our lower bound from \eqref{e:lb} is strictly positive if
\[
\ln m >  \ln \E K_1- \E\ln K_1 \geq \ln\left(\frac{\E K_1}{(\E \sqrt{K_1})^2}\right).
\]
Note that for persistence, we need $m>1$, while for a lower bound on population size we need $$m>\frac{\E K_1}{(\E \sqrt{K_1})^2}= \frac{\Var(\sqrt{K_1})} {[\E(\sqrt{K_1})]^2}+1.$$
This shows that the growth rate $m$ must be large enough to counterbalance fluctuations in carrying capacity, captured by the term $\frac{\E K_1}{(\E \sqrt{K_1})^2}$. These computations are helpful if one wants to establish strictly positive lower bounds for the population size at stationarity, as one does not have an explicit expression for $\E X_\infty$. This could help conservation ecologists who want to set conservation targets.

\subsubsection{Ricker} \label{s:rick}
The Ricker model was introduced to model fish stock in fisheries \citep{ricker1954stock}.  One can also see the Ricker dynamics as the discrete version of Lotka-Volterra dynamics \citep{hofbauer1987coexistence}. 

A single-species stochastic model governed by the Ricker functional response is given by
\begin{equation}\label{e:ricker_1}
X_{t+1} = X_te^{r\left(1-\frac{X_t}{K_{t}}\right)} 
\end{equation}
where the  carrying capacities $(K_t)$ are assumed to form an iid sequence and the growth rate parameter $r$ is supposed to be constant. Here we assumed that the stochasticity affects the carrying capacity, which can be seen as stochasticity in the intensity of competition or the availability of resources.

In the absence of noise, set  $\bar K= \E K_1$. The deterministic dynamics is
\[
\bar X_{t+1} = \bar X_te^{r\left(1-\frac{\bar X_t}{\bar K} \right)} 
\]
If $r\in (0,2)$ it is well known that as $t\to\infty$ one has $\bar X_t\to \bar K$. 

While a discrete logistic equation \citep{streipert2022derivation} of the form
\[
X_{t+1} =X_t + rX_t\left(1-\frac{X_t}{K}\right)
\]
is problematic because it can lead to negative population sizes, this problem does not appear for the Ricker dynamics.

Using Theorem \ref{t:pers1_disc} or by the results cited in \cite{S12} one can also analyze the stochastic model. If $ r>1$ we get that as $t\to\infty$ one has $X_t\to X_\infty$ in distribution. Setting the invasion rate to zero and using Jensen's inequality yields

\[
\E X_\infty = \frac{1}{\E\left( \frac{1}{K_1}\right)}<\frac{1}{\frac{1}{\E K_1}} = \E K_1=\bar K.
\]
\textbf{Biological interpretation:} \textit{This shows that the analogue of the Cushing-Henson conjecture holds in this setting when the growth rate $r\in (1,2)$}: if the environmental fluctuations affect the carrying capacity, they are always detrimental in the Ricker model and lead to a decrease of the population size.

Instead, we can assume that environmental fluctuations affect the growth rate. The stochastic model becomes
\begin{equation}\label{e:ricker_2}
X_{t+1} = X_te^{r_t\left(1-\frac{X_t}{K}\right)}. 
\end{equation}
The deterministic model we can compare it with is
\begin{equation}\label{e:ricker_det_2}
\bar X_{t+1} =  \bar X_te^{\bar r\left(1-\frac{\bar X_t}{ K} \right)} 
\end{equation}
where $\bar r = \E r_1 \in (1,2)$. Then 
\[
\lim_{t\to \infty} \bar X_{t}= K.
\]
In this setting we get
\[
\E\left[ r_1\left(1-\frac{\bar X_\infty}{ K}\right)\right]=0
\]
which implies
\[
E X_\infty = K.
\]
\textbf{Biological interpretation:} \textit{If the environmental fluctuations only affect the growth rate, then the population size does not change.}

\subsubsection{Hassell} \label{s:hass}

A model describing a population of a single species in a limited environment often displays the following two dynamical behaviors:
\begin{itemize}
    \item [1)] exponential growth when the population is small, and,
    \item [2)] density dependence to reduce the rate of growth when the population is large.
\end{itemize}
The Hassell model \citep{hassell1975density} displays both of these behaviors, while also showing a smooth transition between the regimes, just like the logistic model. 
This model is used especially to model insects \citep{hassell1975density,alstad1995managing}  because as the authors of \cite{alstad1995managing} say `it allows for a wide range of density-dependent effects'.

The dynamics is captured by the difference equation
\begin{equation}\label{e:has}
    N_t=\frac{\alpha N_t}{(1+\bar KN_t)^c}.
\end{equation}
The equation has an asymptotically stable fixed point $\bar N=\frac{\alpha^{1/c}-1}{\bar K}>0$ if and only if $\alpha>1$ and $0<c\leq 1$. Under these assumptions, as $t\to \infty$ one has $N_t\to \bar N$. 

Consider the following stochastic version of the Hassell model
\begin{equation}\label{e:has_stoc}
N_{t+1} = \frac{\alpha N_t}{(1+K_t N_t)^c} 
\end{equation}
where we assume that the $K_t$'s form an iid sequence. Let $\E K_1= \bar K$ and suppose that
\[
\alpha>1, ~0<c\leq 1.
\]

Note that in this model, given by \eqref{e:has}, the carrying capacity is not $\bar K$, as in the Beverton-Holt model, but $\frac{\alpha^{1/c}-1}{\bar K}$. As a result, it means the environmental fluctuations, since we keep $\alpha^{1/c}-1$ constant, are influencing the inverse of the carrying capacity. The quantity $\bar K$ can also be interpreted as the per-capita intracompetition rate.

Using Theorem \ref{t:pers1_disc} or the results from \cite{S12} one gets that as $t\to \infty$ we have $N_t\to N_\infty$ in distribution. By Proposition \ref{p:0} the invasion rate at stationarity will be zero
\[
\E \ln\left[\frac{\alpha }{(1+K_1 N_\infty)^c} \right] = 0.
\]
As a result, using Jensen's inequality
\[
\ln \left[\alpha^{\frac{1}{c}}\right] = \E \ln (1+ K_1 N_\infty)< \ln (1 + \E K_1 \E N_\infty) =  \ln (1 + \bar K \E N_\infty) 
\]
which implies that
\[
\bar N = \frac{(\alpha)^{\frac{1}{c}}-1}{\bar K} < \E N_\infty.
\]
\textbf{Biological interpretation:} The last inequality shows that in this setting environmental fluctuations actually help, leading to an increase of the expected population size. The Hassell model does not satisfy the Cushing-Henson conjecture. This is very interesting as it shows that the nonlinearity of this functional response can lead to counterintuitive dynamical behavior under the influence of environmental fluctuations.
\begin{rmk}
 Note that if $c=1$ we recover the Beverton-Holt model and we get an apparent contradiction. This is because in the Beverton-Holt model we have shown that the Cushing-Henson conjecture holds and fluctuations are detrimental while in the Hassell model with $c=1$ we get a complete reversal, namely, that fluctuations lead to an increase in population size. This is because we compare the stochastic version of the Beverton-Holt to the one where the carrying capacity is fixed and equal in average to the deterministic one, while in the Hassell model we compare the stochastic version to the one where the inverse of the carrying capacity is fixed and equal in average to the inverse of the deterministic one. 

What this example shows is that not only do things change if we introduce noise differently in a model, but it also matters what one takes as the deterministic `no-noise' baseline for comparison: do we hold $K_t$ constant at $\E K_1$, or $1/K_t$ constant at $\E (1/K_1)$? The first option means that the carrying capacity is held constant at its arithmetic mean, while the second option says that the carrying capacity is held constant at its harmonic mean.  

See Remark \ref{r:ches} for more details.

\end{rmk}

\subsection{Two species dynamics}\label{s:2d}
Most ecosystems have more than one species, and when one has multiple species there are usually various interactions between them. We will analyze the dynamics of two interacting species. This will then allow us to get some interesting conclusions about the population sizes at stationarity by setting once again $r_1(\mu)=r_2(\mu)=0$ for a coexistence measure $\mu$.
In this setting \eqref{e:SDE} becomes
\begin{equation}\label{e:2d_sys}
\begin{aligned}
X_1(t+1)&=X_1(t) f_1(X_1(t), X_2(t),\xi(t)),\\
X_2(t+1)&=X_2(t) f_2(X_1(t), X_2(t),\xi(t)).
\end{aligned}
\end{equation}
The classification of the long-term behavior proceeds as follows (see \cite{CE89, hening2021coexistence, HNC21}). We assume for simplicity that the noise is `nice' so that invariant probability measures exist (see \cite{hening2021coexistence} for more details) - this will always be possible if we make the noise smooth enough. First one has to look at the Dirac delta measure $\delta_0$ at the origin $(0,0)$, which is the measure where both species are extinct. It is easy to see that
\[
r_i(\delta_0) = \E[\ln f_i(0,\xi_1)], i=1,2.
\]
If $r_i(\delta_0)>0$ then by Theorem \ref{t:pers1_disc} species $i$ survives on its own and converges to a unique invariant probability measure $\mu_i$ supported on $\Se_+^i := \{\bx\in\Se~|~x_i\neq 0, x_j=0, i\neq j\}$. The realized per-capita growth rates can be computed as
  \[
  r_i(\mu_j)=\int_{(0,\infty)}\E[\ln f_i(x,\xi_1)]\mu_j(dx).
  \]

We note that extinction only happens asymptotically as $t\to\infty$ and persistence means the convergence to an invariant probability measure as in Theorem \ref{t:pers1_disc} above.
\begin{enumerate}[label=(\roman*)]
  \item Suppose $r_1(\delta_0)>0, r_2(\delta_0)>0$. This would be the setting where there are two competing species which both survive in the absence of the other species.
  \begin{itemize}
    \item If $r_1(\mu_2)>0$ and $r_2(\mu_1)>0$ we have coexistence and convergence of the distribution of $\BX(t)$ to the unique invariant probability measure $\pi$ on $\Se_+$.
    \item If $r_1(\mu_2)>0$ and $r_2(\mu_1)<0$ we have the persistence of $X_1$ with $X_1(t)\to X_1(\infty)$ weakly as $t\to\infty$ where $X_1(\infty)$ has distribution given by $\mu_1$ and the extinction of $X_2$, that is, $\PP(\lim_{t\to \infty}X_2(t)=0)=1$.
    \item If $r_1(\mu_2)<0$ and $r_2(\mu_1)>0$ we have the persistence of $X_2$ and extinction of $X_1$.
    \item If $r_1(\mu_2)<0$ and $r_2(\mu_1)<0$ we have that for any $\BX(0)=\bx\in\Se_+$
    \[
    p_{\bx,1}+ p_{\bx,2}=1,
    \]
    where $p_{\bx,j}>0$ is the probability that species $j$ persists and species $i\neq j$ goes extinct.
  \end{itemize}
  \item Suppose $r_1(\delta_0)>0, r_2(\delta_0)<0$. Then species $1$ survives on its own and converges to its unique invariant probability measure $\mu_1$ on $\Se^1_+$. In this setting one can see species $1$ would be the prey or a competitor which persists on its own and species $2$ either as a predator or as a second prey species which cannot persist even in isolation.
  \begin{itemize}
    \item If $r_2(\mu_1)>0$ we have the persistence of both species and convergence of the distribution of $\BX(t)$ to the unique invariant probability measure $\pi$ on $\Se_+$.
    \item If $r_2(\mu_1)<0$ we have the persistence of $X_1$ and the extinction of $X_2$.
  \end{itemize}
\item Suppose $r_1(\delta_0)<0, r_2(\delta_0)<0$. Then both species go extinct with probability one.
\end{enumerate}
We will be using this classification in order to say something about the population sizes at stationarity, just like we did in the single-species setting.

\subsubsection{Discrete Lotka-Volterra (Ricker)}

We look at $n$ species which interact according to the system
\begin{equation}\label{e:LVd}
X_i(t+1)=X_i(t)\exp\left(b_i(t)+\sum_j
a_{ij}(t)X_j(t)\right), i=1,2,\dots,n.
\end{equation}
Here we assume that $(b_i, a_{ij})$ are stochastic and form an iid sequence. The coefficients have the following meanings: $a_{ii}$ is the intracompetition coefficient and we assume $a_{ii}<0$ almost surely; $a_{ij}$ quantifies how species $j$ interacts with species $i$ and can capture either competition or predation. Note that we cannot treat mutualism as this can lead to blow-up. Conditions which ensure that \eqref{e:LVd} is well-defined and has nice solutions can be found in \cite{BS19,HNC21}.

These equations are the discrete time analogues of Lotka-Volterra ODE \citep{HHJ87}. According to \cite{C94} one can see that the regular coexistence mechanisms of relative nonlinearity and the storage effect cannot happen in these models due to the fact that the logarithm of the growth rates is a linear function of the densities of the species. Nevertheless, the stochastic environment makes the process explore all of the state space and this allows one to get more results than in the purely deterministic case.

For a full analysis of stochastic discrete Lotka-Volterra systems see \cite{SBA11, BS19, hening2021coexistence, HNC21, HNS20}. 

The linearity of the system \eqref{e:LVd} and
Proposition \ref{p:0} allows us to compute $r_i(\mu)$ for any ergodic
measure $\mu$. For simplicity we assume there
are only two species, so that
\begin{equation}\label{e:2d_disc}
\begin{aligned}
X_1(t+1)&=X_1(t)e^{b_1(t)+
a_{11}(t)X_1(t)+a_{12}(t)X_2(t)} ,\\
X_2(t+1)&=X_2(t)e^{b_2(t)+
a_{22}(t)X_2(t)+a_{21}(t)X_1(t)}.
\end{aligned}
\end{equation}
We first look at the Dirac delta measure $\delta_0$ at the
origin $(0,0)$. The growth rate is
\[
r_i(\delta_0) = \E \ln b_i(1), ~i=1,2.
\]
If $r_i(\delta_0)>0$ then species $i$ survives on its own and
converges to a unique invariant probability measure $\mu_i$
supported on $\Se_+^i := \{\bx\in\Se~|~x_i\neq 0, x_j=0, i\neq
j\}$. Moreover,
\[
r_i(\mu_i) = \E [b_i(1)] + \E [a_{ii}(1)] \int
x_i\,\mu_i(dx_i)=0
\]
which implies
\[
\int x_i\,\mu_i(dx_i) = \frac{\E [b_i(1)]}{ \E [-a_{ii}(1)]}.
\]
One can use this to compute the per-capita growth rates
\begin{equation}\label{e:growth_2d}
r_i(\mu_j)= \E [b_i(1)] + \E [a_{ij}(1)] \int x_j\,\mu_j(dx_j) =
\E [b_i(1)] + \E [a_{ij}(1)] \frac{\E [b_j(1)]}{ \E
[-a_{jj}(1)]}.
\end{equation}
 
Having the expressions \eqref{e:growth_2d} for
$r_1(\mu_2)$ and $r_2(\mu_1)$ we see when the two species persist. There are two possibilities for persistence.

For the first possibility we have a predator-prey system. Assume 1 is the prey and 2 the predator, so that with probability one $a_{12}<0, a_{21}>0, b_1>0, b_2<0$. There is persistence if the prey persists on its own, i.e. $r_1(\delta_0)>0$ and $r_2(\delta_0)<0$, and the prey can be invaded by the predator $r_2(\mu_1)>0$. 

The other possibility is when the two species compete for resources, so that with probability one $a_{12}<0, a_{21}<0, b_1>0, b_2>0$ and $r_1(\delta_0)>0$ and $r_2(\delta_0)>0$ as well as $r_1(\mu_2)>0, r_2(\mu_1)>0$.

Assume that the system converges to the stationary distribution $\mu_{12}$ which supports both species. Then by Proposition \ref{p:0} we get 
\begin{equation*}\label{e:2d_disc}
\begin{aligned}
r_1(\mu_{12})&=0=\E [b_1(1)] + \E [a_{11}(1)] \E  X_1(\infty)+ \E [a_{12}(1)] \E X_2(\infty) \\
r_2(\mu_{12})&=0=\E [b_2(1)] + \E [a_{22}(1)] \E  X_2(\infty)+ \E [a_{21}(1)] \E  X_1(\infty).
\end{aligned}
\end{equation*}
Assume that $\E[a_{11}(1)]\E[a_{22}(1)]-\E[a_{12}(1)]\E[a_{21}(1)\neq 0$.
Solving this linear system yields the unique solution
\begin{equation}\label{e:LV_pop}
\begin{aligned}
\E  X_1(\infty)&=& \frac{\E[a_{22}(1)]\E[b_{1}(1)] -\E[a_{12}(1)]\E[b_{2}(1)]}{\E[a_{11}(1)]\E[a_{22}(1)]-\E[a_{12}(1)]\E[a_{21}(1)]}\\
\E  X_2(\infty)&=&\frac{\E[a_{11}(1)]\E[b_{2}(1)] -\E[a_{21}(1)]\E[b_{1}(1)]}{\E[a_{11}(1)]\E[a_{22}(1)]-\E[a_{12}(1)]\E[a_{21}(1)]}.
\end{aligned}
\end{equation}
In the deterministic setting, if we look at the system

\begin{equation}\label{e:2d_disc}
\begin{aligned}
X_1(t+1)&=X_1(t)e^{\bar b_1+
\bar a_{11}X_1(t)+\bar a_{12}X_2(t)} ,\\
X_2(t+1)&=X_2(t)e^{\bar b_2+
\bar a_{22} X_2(t)+\bar a_{21}X_1(t)}
\end{aligned}
\end{equation}
where $\bar a_i:= \E a_i(1)$ one can show that under suitable conditions on the coefficients, see \cite{luis2011stability, balreira2014local} for details, the unique globally stable fixed point supporting both species is
\begin{equation}\label{e:LV_pop_det}
\begin{aligned}
\bar X_1&=& \frac{\bar a_{22}\bar b_{1} -\bar a_{12}\bar b_{2}}{a_{11}\bar a_{22}-\bar a_{12}\bar a_{21}}\\
\bar X_2&=& \frac{\bar a_{11}\bar b_{2} -\bar a_{21}\bar b_{1}}{a_{11}\bar a_{22}-\bar a_{12}\bar a_{21}}.
\end{aligned}
\end{equation}
As a result, due to the linearity of the model, the total population sizes do not change at all due to the environmental fluctuations since by \eqref{e:LV_pop} and \eqref{e:LV_pop_det}
\[
\E X_1(\infty) = \bar X_1, ~~\E X_2(\infty) = \bar X_2.
\]
\begin{rmk}\label{r:ches}
One might be confused because here we have shown that the expected population size does not change due to stochasticity while in Section \ref{s:rick} we saw that the stochastic fluctuations of the carrying capacity decrease the expected population size. This was already discovered by \cite{chesson1991stochastic}, which we now follow to give an explanation. The difference is between
\[
X_{n+1} = X_ne^{r\left(1-\frac{X_n}{K_{n}}\right)} 
\]
and 
\[
 X_{n+1} =  X_ne^{r\left(1-\alpha_n X_n\right)} 
\]
where $\alpha_n=\frac{1}{K_n}$ is the intraspecific competition coefficient. It depends if we compare the stochastic model with the deterministic model with $\bar K : = \E K_1$ or with the deterministic model with $\bar \alpha := \E \alpha_1 = \E \left(\frac{1}{K_1}\right)$. This is just an apparent contradiction, because what we are actually doing is we are comparing the stochastic models with different deterministic models, and that is why we get different results. Nevertheless, this example  shows that a certain ambiguity can arise and one needs to make a case as to why we should compare the stochastic system to a specific deterministic one, instead of another one \citep{chesson1991stochastic}.
\end{rmk}

\subsection{A predator-prey model}\label{s:PP}

We next analyze a predator-prey model from \cite{streipert2022derivation}. This is an interesting model that was derived from first principles -- it exhibits a rich dynamics even in the deterministic setting. The type of model described in \cite{streipert2022derivation} comes from trying to find meaningful discretizations of continuous time dynamics. If one looks at a single-species model the idea is to express the population size at time $t+1$ as
\[
X_{t+1} = f(t)X_t = \frac{1+p(t)}{1+q(t)}X_t
\]
where $p(t)$ captures what leads to an increase in the population and $q(t)$ what contributes to the decrease of the population. Suppose we have a predator-prey system and $X_t$ is the prey while $Y_t$ is the predator. The authors of \cite{streipert2022derivation} explain why they pick $p(t)=r$ to be the growth contribution while the decline in the population is modeled by the term $q(t)=\frac{r}{K}X_t+Y_t$, which captures competition for resources and predation. The predator population declines at the constant rate $d>0$ so $q(t)=d$ and increases at the rate $p(t)=\gamma X_t$ which is due to predation and $\gamma>0$ is the prey consumption-energy rate of the predator. 

We get
\begin{equation*}\label{e:PP}
\begin{aligned}
X(t+1)&=\frac{(1+r(t))X(t)}{1+\frac{r(t)}{K}X(t)+\alpha Y(t)} \\
Y(t+1)&=\frac{(1+\gamma(t) X(t))Y(t)}{1+d} 
\end{aligned}
\end{equation*}
where $(r(t),\gamma(t))$ form an iid sequence and $(r(1),\gamma(1))$ has mean $(\bar r, \bar \gamma)$. In order for the prey to survive on its own we need
\[
r_X(\delta_0) = \E \ln(1+r(1))>0,
\]
which always holds. Since the prey, in the absence of the predator, has a Beverton-Holt stochastic dynamics, by Theorem \ref{t:pers1_disc} there is a unique ergodic probability measure $\mu_X$ on the positive $x$-axis and $X(t)\to \mu_X$ in distribution as $t\to\infty$. 

At stationarity we get by Proposition \ref{p:0} that
\[
r_X(\mu_X)=0= \E\ln\frac{(1+r(1))X}{1+\frac{r(1)}{K}X}
\]
For the predator to persist we need according the discussion in Section \ref{s:2d} to have
\[
\int \E \ln \frac{(1+\gamma(1) x)}{1+d} \mu_X(dx)>0.
\]
Then, if the noise is `nice' enough (see Section \ref{s:2d}), there will be a unique invariant probability measure $\mu_{X,Y}$ on $\R_{++}^2$ and $(X(t),Y(t))\to (X_\infty, Y_\infty)$ in distribution as $t\to\infty$. Here we took $(X_\infty, Y_\infty)$ to have a distribution given by $\mu_{X,Y}$.
At stationarity when both prey and predator are present we note that by Proposition \ref{p:0} and by Jensen's inequality one gets
\[
r_Y(\mu_{X,Y}) = 0 = \E \ln \frac{1+\gamma(1)X_\infty}{1+d} < \ln \E  \frac{1+\gamma(1)X_\infty}{1+d} = \ln \frac{1+\bar \gamma \bar X_\infty}{1+d}
\]
which implies that
\begin{equation}\label{e:PP2}
\bar X_\infty > \frac{d}{\bar \gamma}.
\end{equation}

The deterministic dynamics is given by
\begin{equation*}\label{e:PP_det}
\begin{aligned}
\bar X(t+1)&=\frac{(1+\bar r \bar X(t)}{1+\frac{\bar r}{K}\bar X(t)+\alpha \bar Y(t)} \\
\bar Y(t+1)&=\frac{(1+\bar \gamma \bar X(t))\bar Y(t)}{1+d}. 
\end{aligned}
\end{equation*}
\cite{streipert2022derivation} showed that if $d<\bar \gamma K< 1+2d$ then the fixed point
\[
E^* =\left(\frac{d}{\bar \gamma},\frac{\bar r(\bar \gamma K-d)}{\alpha \gamma K}\right)
\]
is locally asymptotically stable. Note that \cite{streipert2022derivation} conjectured that $E^*$ is actually globally asymptotically stable. As a result, as $t\to\infty$, at the very least locally around  $E^*$, one has
\begin{equation}\label{e:PP3}
   \bar X(t) \to \frac{d}{\bar \gamma}. 
\end{equation}
\textbf{Biological interpretation:} We see from \eqref{e:PP3} and \eqref{e:PP2} that noise seems to increase the population of the prey in this situation -- see also our simulation results from Figure \ref{f:BH3}. This has been observed in some studies like \cite{abrahams2007predator} where the authors show that environmental fluctuations can sometimes shift the balance in favor of the prey. This could lead to a loss in biodiversity as species that are higher on the trophic chains will become threatened.

\section{Small noise results}\label{s:small_noise}

As we have seen in the previous sections, even in the single-species setting it is usually hard or impossible to get information about the invariant measure $\mu$ which describes the persistence of the ecosystem. The general theory can tell us that $\mu$ exists and is unique but we dont have any formulas describing this measure, and even computing the population sizes at stationarity is usually not possible. In order to circumvent this problem, one can use small perturbations of a deterministic system which has `nice' fixed points. 
Assume we have the deterministic dynamics given by

\begin{equation}\label{e:SDE_det}
\bar X(t+1) = F(\bar \BX(t), \bar e).
\end{equation}

We perturb the dynamics by letting the $(e(t))_{t\geq 0}$ be iid random variables taking values in a convex compact subset of $\R^\ell$. The stochastic dynamics is
\begin{equation}\label{e:gen2}
X_i(t+1) = F_i(\BX(t), e(t)).
\end{equation}
The next remark provides some intuition and heuristics. 
\begin{rmk}[Sketch of main ideas] Assume that, using  the theory of stochastic persistence one can show that the species coexist and converge to a unique invariant probability measure (stationary distribution) $\mu$. There are few tools that can give us the analytical properties of the stationary distribution. This makes it hard to quantify how noise changes the population size at stationarity, since this is a functional of the stationary distribution, i.e. the $i$th species has the total population size $\E_\mu X_i=\int x_i \mu(dx)$. We were able to sidestep this issue by using the approximation methods developed by \cite{stenflo1998ergodic, cuello2019persistence}. Cuello's results \citep{cuello2019persistence} lead to a Taylor-type approximation which says loosely that if the noise is small, that is $e=\bar e + \rho \xi$ for $0<\rho \ll 1$, and the deterministic dynamics $X(t+1) = F(X(t),\bar e)$ has a stable equilibrium $\bar X$ then for initial conditions $X(0)$ close to that equilibrium the dynamics converges as $t\to\infty$ to a stationary distribution $\mu$ that lives `around' $\bar X$. If $\bar X$ is globally stable then one can show that the convergence holds for any initial conditions. \cite{cuello2019persistence} then uses these results to prove that $\E_\mu X_i$ has a Taylor-type expansion that allows one to compute all the terms explicitly as a function of $F$ and its various derivatives, including derivatives with respect to noise terms, evaluated at $(\bar X, \bar e).$  
\end{rmk}

Define $vec(\cdot)$ to be the operation that concatenates the columns of a matrix into a vector and let $\otimes$ denote the Kronecker product of matrices.

Using the results from \cite{stenflo1998ergodic} and especially from \cite{cuello2019persistence} (see Appendix \ref{a:small_noise} for details and assumptions) in our stetting one gets the following.

\begin{thm}\label{t:small_noise}
Suppose $e_i=\bar e_i + \rho \xi(t)$ where $(\xi(t))_{t\geq 0}$ are iid random variables taking values in a convex compact subset of $\R^\ell$ and $\rho>0$ is a small parameter. Suppose the deterministic system \eqref{e:SDE_det} has a globally attracting fixed point $\bar \NN \in \R_{++}^n$. Let $D_x F$ be the Jacobian at $ \bar \NN$ with respect to the population i.e. the $i-j$th entry of $D_x F$ is $\frac{\partial F_i}{\partial x_j}(\bar \NN)$. Similarly $D_e F$ will be the matrix with entries $\diag(\frac{\partial F_i}{\partial e_j}(\bar \NN))$.
We have that for $\rho$ small enough, for any $\BX(0)\in \R_{++}^n$, with probability one, the occupation measures of the process $\BX(t)$ given by \eqref{e:gen2} converge weakly as $t\to\infty$ to an invariant probability measure $\mu$ supported around $\bar \NN$. Moreover, if we let $\bar \NN^\rho$ be a vector with distribution $\mu$ we have
\[
\E \bar N_i^{\rho} =  \bar N_i + \sum_{l=1}^n\sum_{j_1, j_2=1}^n b_{il}\frac{\partial^2 F_l}{\partial x_{j_1}\partial x_{j_2}}A_{j_1,j_2} +  \sum_{l=1}^n\sum_{j_1, j_2=1}^\ell b_{il}\frac{\partial^2 F_l}{\partial e_{j_1}\partial e_{j_2}} Cov(e_{j_1}, e_{j_2}) + O(\rho^3), ~~i=1,\dots,n
 \]
where $(b_{il})$ are the entries of the matrix $(I-D_xF( \bar \NN))^{-1}$ and $(A_{ij})$ the entries of the matrix defined by
\[
vec(A) = (I-D_xF \otimes D_xF)^{-1} (D_eF \otimes D_eF) vec(Cov(e))
\]
The change in expected population size of species $i$ because of noise is therefore given by
\begin{equation}\label{2ordApp}
\E \bar N_i^{\rho} -  N_i= \sum_{i,l=1}^n\sum_{j_1, j_2=1}^n b_{il}\frac{\partial^2 F_l}{\partial x_{j_1}\partial x_{j_2}}A_{j_1,j_2} +  \sum_{i,l=1}^n\sum_{j_1, j_2=1}^\ell b_{il}\frac{\partial^2 F_l}{\partial e_{j_1}\partial e_{j_2}} Cov(e_{j_1}, e_{j_2}) + O(\rho^3).    
\end{equation}

\end{thm}

\begin{rmk}
Since we have the relation $e_i=\bar e_i + \rho\xi(t)$ it follows that 
$$(Cov(e))=\rho^2 Cov(\xi(t))   $$
and 
$$Cov(e_{j_1},e_{j_2})=\E((e_{j_1}-\bar e_{j_1})(e_{j_2}-\bar e_{j_2}))=\rho^2\E(\xi_{j_1}(t)\xi_{j_2}(t)).$$
If we define 
\[
vec(\tilde{A})=(I-D_xF \otimes D_xF)^{-1} (D_eF \otimes D_eF) vec(Cov(\xi(t)))
\]
the above show that we can write \eqref{2ordApp} as 
\begin{equation}\label{2ordApp1}
\E \bar N_i^{\rho} -  N_i= \left(\sum_{i,l=1}^n\sum_{j_1, j_2=1}^n b_{il}\frac{\partial^2 F_l}{\partial x_{j_1}\partial x_{j_2}}\tilde{A}_{j_1,j_2} +  \sum_{i,l=1}^n\sum_{j_1, j_2=1}^\ell b_{il}\frac{\partial^2 F_l}{\partial e_{j_1}\partial e_{j_2}} Cov(\xi_{j_1}(t),\xi_{j_2}(t))\right)\rho^2 + O(\rho^3).    
\end{equation}    
\end{rmk}

\textbf{Biological interpretation:} If one perturbs a deterministic ecological system that has a globally attracting fixed point $\bar \NN$ by small noise which fluctuates like $\rho \xi(t)$ around the mean $\bar e$ for a small $\rho>0$, then the stochastic system converges to an invariant probability measure supported around $\NN$. Furthermore, the expectation of the $i$th population size at stationarity looks like $\bar N_i$ plus a term of order $\rho^2$. This $\rho^2$ term  depends on the first and second derivatives with respect to the population sizes and random parameters of the vector field evaluated at the deterministic fixed point $(\bar \NN, \bar e)$ as well as on the covariance of the noise terms. This shows that one can do a Taylor-type expansion around the deterministic fixed point and that, as expected, small environmental fluctuations change the dynamics in a `smooth' fashion. 

Another consequence of \eqref{2ordApp1} is that the change in mean is $O(\rho^3)$ while the fluctuations about te mean are $O(\rho)$. This shows that for small random perturbations, the changes in the mean are dominated by the stochastic temporal variation. Nevertheless, small-noise approximations are useful because they are often accurate even for moderate perturbations.

\begin{rmk}
 We note that small noise approximations have been used for quite some time already. These approximations are, for example, similar to the $\delta$ method from statistics \citep{oehlert1992note}. M. Bartlett has used related small noise approximations for various epidemic and population dynamics models \citep{bartlett1956deterministic, bartlett1957theoretical}. In fisheries models one can also find small noise approximations \citep{horwood1983general, getz1984production}.
\end{rmk}

We get as immediate corollaries the following two results for single-species systems.

\begin{cor}\label{c:1s_1n} (Single Species and 1D Noise)  Suppose we have a single-species ecosystem and only one coefficient is random. Suppose further that $\bar N_1$ is the deterministic stable equilibrium. Then the conclusions of Theorem \ref{t:small_noise} hold and
\[
\E \bar N_1^{\rho} =  \bar N_1 + b_1\frac{\partial^2 F_1}{\partial x_{1}^2}\tilde{A}_1\rho^2 +  b_1\frac{\partial^2 F_1}{\partial \xi_{1}^2}\E(\xi_1^2(t))\rho^2 + O(\rho^3)
 \]
where $b_1=\left(1-\frac{\partial F_1}{\partial x_1}(\bar N_1)\right)^{-1}$ and 
$\tilde{A}_1=\left(1-\left(\frac{\partial F_1}{\partial x_1}(\bar 
 N_1)\right)^2\right)^{-1}\left(\frac{\partial F_1}{\partial \xi_1}(\bar N_1)\right)^2\E(\xi_1^2(t)).$
\end{cor} 

If a deterministic ecological system has a fixed point $\bar \NN$ such that $\bar \NN=(\bar N_1,0,0...,0)$, $\xi(t)$ lies in $\R$, and the fixed point is stable in $\R\times 0$, then the ecological system and its random perturbation are reduced to the simple $n=1$, $\ell=1$ case, where the above corollary applies.

In addition to Corollary \ref{c:1s_1n}, it is relevant to have more than one source of noise to describe a stochastic ecological system. It then becomes important to quantify how the correlation of the noise terms affects the expected population size.
\begin{cor}\label{c:1s_2n} (Single Species and 2D Noise) Suppose we have a single species ecosystem and two coefficients are random. If $\overline{N}$ is the stable equilibrium of \eqref{e:SDE_det}, and $(\xi(t))=(\xi_1(t),\xi_2(t))$ is the noise then the conclusions Theorem \ref{t:small_noise} hold and
\[
\E( \bar N_1^{\rho})=\bar  N_1+b\frac{\partial^2 F}{\partial x_1^2}\tilde{A}\rho^2+b\left(\frac{\partial^2 F}{\partial \xi_1^2}\E(\xi_1^2)+ 2\frac{\partial^2 F}{\partial \xi_1 \partial\xi_2}\E(\xi_1\xi_2)+\frac{\partial^2 F}{\partial \xi_2^2}\E(\xi_2^2)\right)\rho^2+O(\rho^3)
\]
where $b=\left(1-\frac{\partial F}{\partial N}(\overline{N})\right)^{-1}$ 
and 
$$\tilde{A}=\left(1-\left(\frac{\partial F}{\partial x_1}(\overline{N}_1)\right)^2\right)^{-1}\left(\left(\frac{\partial F}{\partial \xi_1}\right)^2\E(\xi_1^2)+\left(\frac{\partial F}{\partial \xi_1}\right)\left(\frac{\partial F}{\partial \xi_2}\right)\E(\xi_1\xi_2)+\left(\frac{\partial F}{\partial \xi_2}\right)^2\E(\xi_2^2)\right).$$
\end{cor}

Each of the examples below has been chosen because it showcases some interesting behavior. Using small-noise approximations we can get explicit expressions of how noise influences the average population size. Not only that, but we can see what happens when there are multiple parameters, like the growth rate and the carrying capacity, which are influenced by environmental fluctuations. Here we see how the correlations play an important role. Here we see how the correlations play an important role. This sometimes leads to different behaviors than those we were able to show in the large noise results from Section \ref{s:large_noise}.

\subsection{1D Beverton-Holt model}

Consider the stochastic Beverton-Holt model

\[
X(t+1)=\frac{r(\xi(t))X(t)}{1+\left(\frac{r(\xi(t))-1}{K(\xi(t))}\right)X(t)}
\]
where both the growth rate $r(\xi_1(t))=\overline{r}+\rho\xi_1(t)$ and the carrying capacity $K(\xi_2(t))=\overline{K}+\rho\xi_2(t)$ are influenced by random environmental fluctuations. Here
$(\xi(t))=(\xi_1(t),\xi_2(t))$ are iid random variables with $\E(\xi_1(t))=\E(\xi_2(t))=0$, $\Cov(\xi_1,\xi_2)=\E(\xi_1\xi_2)$ and $\E(\xi_i^2)=1$.

We assume that the deterministic model with coefficients $\bar r, \bar K$ has a positive stable equilibrium, something which is guaranteed when $\bar r>1$.

Using \eqref{e:SDE_det} with 
$$F(x,\xi_1,\xi_2)= \frac{r(\xi_1)x}{1+\left(\frac{r(\xi_1)-1}{K(\xi_2)}\right)x}            $$
together with Corollary \ref{c:1s_2n} we get that
\[
\E \bar N^\rho(\infty)=\overline{K}+ \left(-\frac{2}{\bar K\bar r}\left(\left(\frac{\bar r-1}{\bar r+1}\right) + 1\right)+\frac{2}{\bar r(\bar r-1)}\E(\xi_1\xi_2)\right)\rho^2 +O(\rho^3)
\]

\textbf{Biological interpretation:} We can see that if the fluctuations of the carrying capacity and the intrinsic growth rate are sufficiently strongly positively correlated, then the expected population size can exceed the deterministic population at equilibrium. Compare this to Section \ref{s:B_H} and the Cushing-Henson conjecture, in which it was shown that environmental fluctuations in carrying capacity are always detrimental.

\subsection{Hassell model }
The deterministic Hassell model is given by
\begin{equation}\label{hassell}
    X(t+1)=\frac{\overline{\alpha} X(t)}{(1+KX(t))^c}
\end{equation}
and has the nontrivial fixed point 
$$\frac{\overline{\alpha}^{1/c}-1}{K}.$$

Consider the quantity $\Tilde{c}=c(1-\overline{\alpha}^{-1/c})$. According to \cite{hassell1975density}, the fixed point is stable if $0<\Tilde{c}< 2$. The population decreases exponentially to the fixed point when $0<\Tilde{c}<1$, and for $1<\Tilde{c}<2$ the population shows damped oscillations to the fixed point. Note that $\overline{\alpha}$ is a parameter in the deterministic Hassell model. So $\tilde{c}$ depends on $c$ and $\overline{\alpha}$.

Let $\tau(t)=(\tau_1(t),\tau_2(t))_{t=0}^{\infty}$ be i.i.d. with $\E(\tau(t))=0$ and $\E(\tau_i^2(t))=1$. The stochastic Hassell model is given by
\begin{equation}\label{hass_rdm}
    X(t+1)=\frac{\alpha(t) X(t)}{(1+K(t)X(t))^c}
\end{equation}
where $K(t)=\overline{K}+\rho\tau_1(t)$ and $\alpha(t)=\overline{\alpha}+\rho\tau_2(t).$
Consider for the sake of simplicity that $\tilde{c}=\frac{1}{2}$ which gives us $\bar\alpha=(\frac{2c}{2c-1})^c$, and $c\in (1/2,1)$.

\textbf{Case 1 -- one-dimensional noise (only $K$ is random)}: Suppose in this setting that $\alpha(t)=\overline{\alpha}$ is fixed.
Using Corollary \ref{c:1s_1n} we get
\begin{equation}
\E(X(\infty))=\frac{1}{\bar K(2c-1)}-\frac{3c-1}{6c(2c-1)K^3}\rho^2+\frac{c+1}{2c(2c-1)K^3}\rho^2+O(\rho^3).
\end{equation}

\textbf{Biological Interpretation}: Since we have $c\in (1/2,1)$, we can see that the expected population exceeds the population at equilibrium for small noise, which agrees with the conclusion we got in Subsection \ref{s:hass}.

\textbf{Case 2 -- two-dimensional noise ($\alpha$ and $K$ are random)}:
The small-noise approximation from Corollary \ref{c:1s_2n} implies that
\begin{equation}
\begin{split}
 \E(X(\infty))&=\frac{1}{\bar K(2c-1)}-\frac{2(3c-1)}{3c(2c-1)\bar \alpha^2\bar K}\rho^2-\frac{(3c-1)}{6c(2c-1)\bar K^3}\rho^2+\frac{c+1}{2c(2c-1)\bar K^3}\rho^2 \\
 &\left(\frac{3c-1}{3c(2c-1)\alpha\bar K^2}-\frac{2}{(2c-1)\bar \alpha \bar K^2}\right)\rho^2\E(\tau_1\tau_2)+O(\rho^3).
\end{split}
\end{equation}

\textbf{Biological interpretation:} From the above expression we observe that if $\bar\alpha<\sqrt{2(3c-1)}K$, and $\E(\tau_1\tau_2)\geq 0$, then the expected population is less than the population at equilibrium. If  $\bar\alpha<\sqrt{2(3c-1)}K$, and $\E(\tau_1\tau_2)\leq 0$, then the expected population size exceeds the deterministic population at equilibrium. Similarly, if instead  $\bar\alpha>\sqrt{2(3c-1)}K$ and $\E(\tau_1\tau_2)\geq 0$ then the environmental fluctuations lead to an increase of the population size. If multiple parameters are fluctuating we get a richer dynamical behavior.

\subsection{A predator-prey model}

We look again at the predator-prey model from \cite{streipert2022derivation}

\begin{equation}\label{e:PP}
\begin{aligned}
X(t+1)&=\frac{(1+r)X(t)}{1+\frac{r}{K}X(t)+\alpha Y(t)} \\
Y(t+1)&=\frac{(1+\gamma X(t))Y(t)}{1+d} 
\end{aligned}
\end{equation}

We first want to find the fixed points for this system. Clearly from what we have, assume that we have a fixed point, then $(X(t),Y(t))=(X(t+1),Y(t+1))$ for $t\geq 0$.
As a result we get 

$$ \frac{r}{K}X(t)^2+\alpha X(t)Y(t)=r X(t)    $$
$$ d Y(t)=\gamma X(t)Y(t).          $$
Solving this system gives us 3 possible fixed points
$$(X,Y)=(0,0)$$
$$(X,Y)=(K,0)$$
$$(X,Y)=\left(\frac{d}{\gamma}, \frac{1}{\alpha}\left(r-\frac{r d}{K\gamma}\right)\right)$$
where the last fixed point exists if $K\gamma > d$.

In order to write this in the general form of our equation, we have 
\begin{align}
    F(X,Y)=\begin{bmatrix}
        \frac{(1+r)X(t)}{1+\frac{r}{K}X(t)+\alpha Y(t)} &  \frac{(1+\gamma X(t))Y(t)}{1+d}\\    
    \end{bmatrix}. 
\end{align}

Consider 

\begin{align}
    DF(x,y)=\begin{bmatrix}
     \frac{(1+r)(1+\alpha y)}{(1+\frac{r}{K}x+\alpha y)^2}   &  -\frac{\alpha(1+r)x}{(1+\frac{r}{k}x+\alpha y)^2}     \\  
     \frac{\gamma y}{1+d}                              & \frac{1+\gamma x}{1+d}\\
    \end{bmatrix} 
\end{align}
at each of the fixed points for this model. At the origin we get

\begin{align}
    DF(0,0)=\begin{bmatrix}
     (1+r)   &  0     \\  
     0       & \frac{1}{1+d}\\ 
    \end{bmatrix} 
\end{align}
and since the spectral radius of this matrix exceeds 1, this implies that the point $(0,0)$ is not internally stable -- see Appendix \ref{a:small_noise} for the definition of internally stable equilibrium.

For the other fixed point on the boundary we get
\begin{align}
    DF(K,0)=\begin{bmatrix}
     \frac{1}{1+r}   &  -\frac{\alpha K}{(1+r)}    \\  
     0       & \frac{1+\gamma K}{1+d}\\ 
    \end{bmatrix} 
\end{align}
and $(K,0)$ is internally stable. Finally,

\begin{align}
    DF\left(\frac{d}{\gamma}, \frac{1}{\alpha}\left(r-\frac{r d}{K\gamma}\right)\right)=\begin{bmatrix}
     \frac{1+r-\frac{rd}{K\gamma}}{1+r}   &  -\frac{\alpha d}{\gamma(1+r)}    \\  
     \frac{\gamma}{1+d}(\frac{1}{\alpha}(r-\frac{rd}{K\gamma}))       & 1\\ 
    \end{bmatrix} 
\end{align}
If we consider $d< \gamma K < 1+2d$, then the fixed point is indeed internally stable.

\subsection{Predator-Prey Model with Stochasticity}

We want to write the model in the general form \eqref{e:SDE_det}, with 

$$F_1(X(t),Y(t),\xi(t))=\frac{(1+r(\xi_1(t)))X(t)}{1+\frac{r(\xi_1(t))}{K}X(t)+\alpha Y(t)}         $$
$$F_2(X(t),Y(t),\xi(t))=\frac{(1+\gamma(\xi_2(t)) X(t))Y(t)}{1+d}   $$
where we have $r(\xi_1(t))=\overline{r}+\rho\xi_1(t)$, and $\gamma(\xi_2(t))=\overline{\gamma}+\rho\xi_2(t)$.

We let $(\xi_1(t),\xi_2(t))_{t=0}^{\infty}$ to be iid random variables with $\E(\xi(t))=0$, $\E(\xi_i^2(t))=1$ and covariance 
\begin{align}
    \Cov(\xi_1(t), \xi_2(t))=\begin{bmatrix}
     1   &  b     \\  
      b     &   1   \\
    \end{bmatrix} 
\end{align}
Let us first analyze the fixed point $(K,0)$. Using Corollary \ref{c:1s_2n} and noting that
$b_1=\frac{\overline{r}+1}{\overline{r}}$,
$A_1=\frac{\overline{r}(\overline{r}+2)K^2\rho^2}{(\overline{r}+1)^4}$,
$\frac{\partial^2 F_1}{\partial N_1^2}=-\frac{2\overline{r}}{K(1+\overline{r})^2},$
we get
$$\E(X(\infty))= K  + O(\rho^3).   $$

Now, for the fixed point $\left(\frac{d}{\gamma}, \frac{1}{\alpha}\left(r-\frac{r d}{K\gamma}\right)\right)$, we assume $d<\gamma K < 1+2d$, set $b=r-\frac{rd}{K\gamma}$ and introduce the following simplifications
\begin{itemize}
    \item $\frac{\alpha}{\gamma}=p$
    \item $d=r$
    \item $r-\frac{rd}{K\gamma}=b=\frac{r}{2}$
    \item $\Cov(\xi(t))=I$.
    
\end{itemize}

Some routine computations and the use of Theorem \ref{t:small_noise} show that

\begin{equation*}
\begin{split}
\E(X(\infty))+ \E(Y(\infty)&= \frac{d}{\gamma}+ \left(\frac{r(r + 1)}{\gamma(7r^2 + 14r + 8)}\right)\rho^2 +\left(\frac{p^2r^2(2r^2 + 5r + 4)}{\alpha^2\gamma(7r^2 + 14r + 8)}\right)\rho^2 + O(\rho^3)\\
&+\frac{r}{2\alpha}-\left(\frac{r^2(r+2)\rho^2}{Kp\gamma^2(1+r)}\right)\left( \frac{2(r + 1)^2}{r(7r^2 + 14r + 8)}+ \frac{rp^2(4r^2 + 7r + 4)}{\alpha^2(7r^2 + 14r + 8)}\right) \\
 &+\left(\frac{r\rho^2}{\gamma(1+r)}\right)\left(\frac{r + 1}{p(7r^2 + 14r + 8)}+\frac{rp(2r^2 + 5r + 4)}{2\alpha^2(7r^2 + 14r + 8)}\right)\\
 &+\frac{pr^2\rho^2}{2\gamma(1+r)}\left(\frac{(4r^2 + 7r + 4)}{r(7p^2r^2 + 14p^2r + 8p^2)}+\frac{r(8r^2 + 19r + 12)}{4\alpha^2(7r^2 + 14r + 8)}\right)\\
 & -\frac{r\rho^2}{2\gamma}\left(\frac{2(r + 1)^3}{7pr^2 + 14pr + 8p} +\frac{r(r + 1)^2(2pr^2 + 5pr + 4p)}{2\alpha^2(7r^2 + 14r + 8)}\right)-\frac{\rho^2}{2\alpha}+O(\rho^3).
 \end{split}
\end{equation*}

\textbf{Biological Interpretation :} We can see that small noise increases the expected population of the prey, as we have shown in Subsection \ref{s:PP} using different methods. The expected population of the predator has a very complicated and it is not immediate to see under which parameter values it is below or above the deterministic population equilibrium $\frac{r}{2\alpha}$ - however, see Figure \ref{f:BH3}.

\section{Numerical experiments and simulations}\label{s:sim}

\subsection{Simulations}
In this section we present simulation results for some of the models we analyzed in the paper. This is to showcase numerically some of the results which we were able to prove analytically in the previous sections. Moreover, we can also see, by the simulations, some large noise results which we were not able to prove rigorously. These numerical explorations can be done for both small and large environmental fluctuations and therefore present some interesting phenomena.
For all of the simulations we have taken the random variables to be lognormally distributed. In the cases where a pair of lognormal random variables is considered, we define the correlation matrix $\Sigma$, for example if we have $Y=(Y_1,Y_2)$, then the following relations hold 
\[
\E(Y)_i=e^{\mu_i+\frac{1}{2}\Sigma_{ii}}
\]
\[
Var(Y)_{ij}=e^{\mu_i+\mu_j+\frac{1}{2}(\Sigma_{ii}+\Sigma_{jj})}(e^{\Sigma_{ij}}-1)
\]

We approximate the population size $X(\infty)$ at stationarity by the sample means
\[
\hat X(T) = \frac{1}{T}\sum_{t=0}^T X(t)
\]
for $T$ large enough. Note that with probability one $$\hat X(T)\to \E X(\infty)$$ as $T\to  \infty$.

The Figures \ref{f:BH11}, and \ref{f:BH22}, are the plots for the expected population, versus noise strength. The simulations have been performed by taking sample means over $25000$ data points, and the noise parameter $t$ takes $10000$ values evenly distributed between $0$ and $1$. The plots have a zigzag pattern due to in-sample variability, and it can be seen that the jaggedness decreases as the sample size is made significantly larger, and the relation between average population and noise strength is more clearly visible.

\subsection{Single species Beverton-Holt}
For the Beverton-Holt model
$$ X(t+1)= \frac{(s(t)+1)X(t)}{1+\left(\frac{s(t)}{K(t)}\right)X(t)}.$$
one can find the numerical experiments in Figures \ref{f:BH11} and \ref{f:BH22}.

\begin{figure}

    \begin{subfigure}{0.45\textwidth}            
            \includegraphics[width=\textwidth]{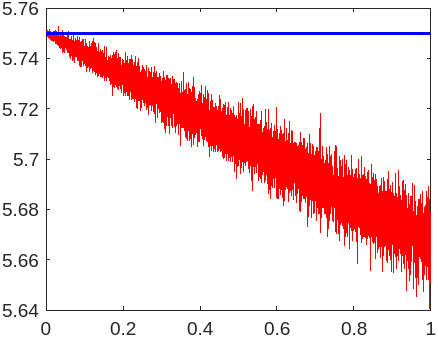}
            \caption{\textbf{$\overline{K}<r+1$}}
           
    \end{subfigure}
    \begin{subfigure}{0.45\textwidth}
            \centering
            \includegraphics[width=\textwidth]{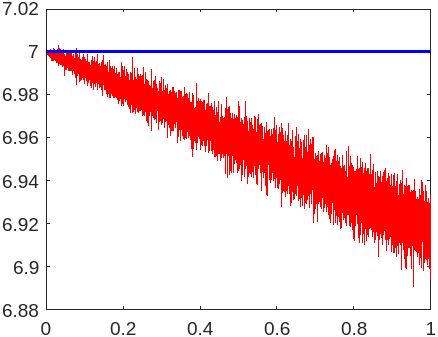}
            \caption{\textbf{$\overline{K}>r+1$}}
            
    \end{subfigure}
    \caption{Expected population size versus noise strength, $t$, in a Beverton-Holt model. We suppose the random variables to be uncorrelated: $Cov((s(t),K(t)))=tI$. \textbf{A}: $K=5.75$ and $r=5$. \textbf{B}: $K=7$ and $r=5$. We see a clear decreasing trend between the expected population and the noise strength parameter. The analytical results show a decrease in the expected population, if the noise is small. The simulation shows the same behavior holds even when the noise parameter is relatively large. }
    \label{f:BH11}
\end{figure}

\begin{figure}
    \includegraphics[width=0.7\linewidth]{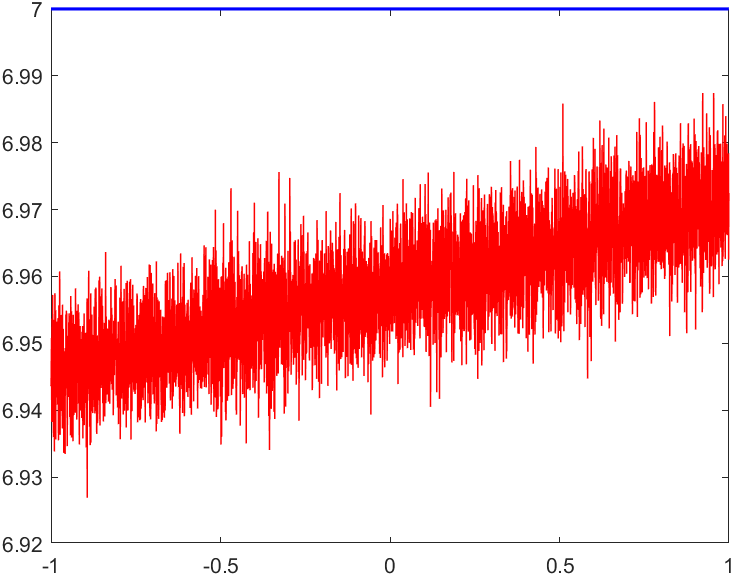}
 \caption{Here we consider the random variables from the single-species Beverton Holt model, with $\bar r=5$, and $\bar K=7$ to be correlated. We fix $\mathbf{E}(r^2(t))= 0.5$, $\mathbf{E}(K^2(t))=0.5$ and $\E(r(t)K(t))=0.25b$
and plot the sample mean of the population versus $b$ which goes from -1 to 1. Since our analytical result for small noise shows that negative correlation between the random variables affects the average population negatively, and positive correlation affects the average population positively. We see that even when noise is not small, the average population follows a trend like the one predicted for small noise. There is an overall decrease in average population due to noise, like we have seen in Figure \ref{f:BH11}, for the uncorrelated case. }
\label{f:BH22}
\end{figure}

\subsection{Predator-Prey Model}  

The numerical experiments for the predator-prey model
\begin{equation}
\begin{aligned}
X(t+1)&=\frac{(1+r(\xi_1(t))X(t)}{1+\frac{r(\xi_1(t))}{K}X(t)+2\overline{\gamma} Y(t)} \\
Y(t+1)&=\frac{(1+\gamma(\xi_2(t)) X(t))Y(t)}{1+2\overline{r}}
\end{aligned}
\end{equation}
can be seen in Figure \ref{f:BH3}.

\begin{figure}

    \begin{subfigure}{0.45\textwidth}            
            \includegraphics[width=\textwidth]{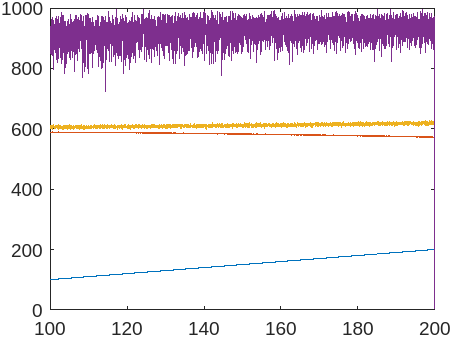}
            \caption{\textbf{Prey Population vs $\overline{r}$}}
           
    \end{subfigure}
    \begin{subfigure}{0.45\textwidth}
            \centering
            \includegraphics[width=\textwidth]{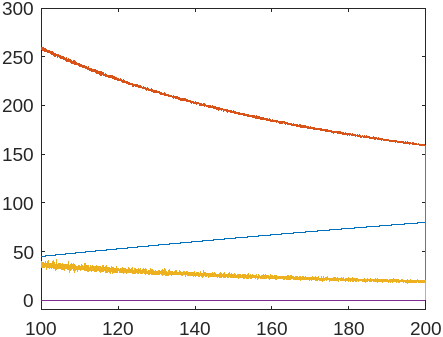}
            \caption{\textbf{Predator Population vs $\overline{r}$}}
            
    \end{subfigure}
    \caption{Prey/Predator population versus $\overline{r}$ in the predator-prey model, with $K=1000$, $\overline{\gamma}=1$, $\alpha=2\overline{\gamma}$, $d=2\overline{r}$. We consider the random variables to be uncorrelated: $\Sigma=diag(100s,s)$. We see the plot for the prey population on the left, and predator population on the right. The blue curves indicate how the equilibrium population changes with $\overline{r}$, and the red curves indicate the sample mean population with $s=5\times 10^{-4}$, the yellow curves with  $s=5\times 10^{-2}$ and the purple curves with $s=5\times 10^{-1}$. The prey population clearly increases as the noise parameter $s$ increases. For the predator population we observe that small noise is beneficial to the predator population, however large noise proves detrimental. Our analytical result only deals with small noise, and does not give a clear picture regarding the expected population while the plots show what could happen when the noise is not small.}
    \label{f:BH3}
\end{figure}

\section{Discussion and future work}\label{s:discussion}

\textbf{Large noise} We first looked at the general framework of stochastic persistence that has been developed recently \citep{BS19, B18, H19, HNC21}. Using this framework, one is able to show when a stochastic ecological system converges to a unique stationary distribution which supports all the species. We were able to prove some interesting results regarding the effects of noise on the stationary population size in various models. 

For one-dimensional models we showed that in many examples (Beverton-Holt, Ricker), as expected, noise decreases the total population size. A more counterintuitive result is that, in the Hassell model, environmental fluctuations lead to an increase of the total population size. The most interesting part is that the results seem to be not only model dependent, but one also gets different behavior based on which model parameters are influenced by environmental stochasticity. In the Ricker model, we get that if the carrying capacity experiences stochasticity then fluctuations are detrimental and the average population size decreases. If instead stochasticity influences the growth rate, then the average population size remains unchanged. We show that, just as \cite{chesson1991stochastic} pointed out, the effect of stochasticity on total population size can be ambiguous because it depends which deterministic model we take as the standard for comparison. One can get different answers by comparing the same stochastic model with different averaged deterministic models. This tells us we have to be careful when we make a comparison - there needs to be a good ecological reason to pick one comparison over another.

\textbf{Small noise.} The small noise results from Section \ref{s:small_noise} allowed us to further explore how noise influences the population size at stationarity. It turns out that by varying more than one parameter in a single-species model can lead to complex terms which can increase or decrease the population size. In the Beverton-Holt model positive correlations between the growth rate and carrying capacity can lead to an increase of the population size, whereas we know by the Cushing-Henson conjecture that if we just vary the carrying capacity the effects of noise are always detrimental. In the Hassell model, if we vary only the carrying capacity then the  population size at stationarity increases. If however, we let both the carrying capacity and the growth be random, we get that a positive correlation between these two random variables sometimes increases the population size and at other times decreases the population size, showcasing once again that very interesting phenomena can occur in the Hassell dynamics. For two-species models we look at a predator-prey model from \cite{streipert2022derivation}. The small noise results show that in the predator-prey model, the prey population increases due to the noise while the behavior for the predator is significantly more complicated.

\textbf{Future work.} Our current results are only for stochastic difference equations. However, many important biological models are continuous in time. It will be interesting to see how the population size changes due to environmental fluctuations for various types of models, like stochastic differential equations \citep{HN16, HNC21}, piecewise deterministic Markov processes \citep{BL16, B18, HN19, hening2021stationary} and stochastic differential equations with switching \citep{hening2021stationary, FS24}. The small noise result of \cite{cuello2019persistence}, which we made extensive use of throughout the paper, cannot be easily generalized to continuous time dynamics. A first step will be to prove small-noise estimates for functionals of the invariant measure around a stable equilibrium for continuous time models.

\textbf{Acknowledgments:} The authors thank Sebastian Schreiber and Peter Chesson as well as the two anonymous referees for helpful suggestions which led to an improved manuscript. A. Hening acknowledges generous
support from the NSF through the CAREER grant DMS-2339000. 

\textbf{Conflict of interest.} The authors have no conflict of interest to declare.

\appendix

\section{Small noise}\label{a:small_noise}

The deterministic dynamics is given by 
\[
\bar X_i(t+1) = F_i(\bar \BX(t), \bar e).
\]

The dynamics is made stochastic by letting $(e(t))_{t\geq 0}$ be iid random variables taking values in a Polish space.
Assume the following.
\begin{asp}\label{a:Cuello2}
One has:
\begin{itemize}
    \item [(\textbf{B1})] $e(t)=\bar e +\rho\xi(t)$ for each $t$ where $(e(t))_{t\geq 0}$ are iid random variables which take values in a convex compact subset of $\R^\ell$ with $\E(\xi(t))=0$ and $\rho\geq 0$.
    \item [(\textbf{B2})] If $x^*\in \Se$ is an equilibrium for $G(x)=F(x,\bar e)$, then $x^*$ is internally hyperbolic. 
    \item [(\textbf{B3})] For every $x\in \Se$, the deterministic dynamics \eqref{e:SDE_det} satisfies $\omega(x)=x^*$ for some equilibrium $x^*\in \Se$, where $\omega(x)$ is the omega-limit set of $x$. 
    \item [(\textbf{B4})] Any subspace of the form $\Se_I=\{x=(x_1,\dots,x_n)\in \R_+^n~|~x_i=0, ~i\notin I \}$ is invariant for the dynamics \eqref{e:SDE_det}. 
\end{itemize}
\end{asp}
The stochastic dynamics is
\begin{equation}\label{e:gen3}
X_i(t+1) = F_i(\BX(t), e(t)).
\end{equation}

Following \cite{cuello2019persistence} we say an equilibrium $x^*\in \Se_I$ is internally stable if all the eigenvalues of the Jacobian matrix restricted to $\Se_I$, i.e. we delete rows and columns from the $n\times n$ Jacobian for all $i$ for which $i\notin I$, have magnitude strictly less than one.

We will make use of some small-noise results by \cite{stenflo1998ergodic} and \cite{cuello2019persistence}.

\begin{thm} (Theorem 3.3.1 from \cite{cuello2019persistence})\label{sfp_stoch2}
Consider a Markov chain that satisfies an equation of type $\eqref{e:gen3}$. If Assumption \textbf{B1} above holds and the equilibrium $x^*\in S_{I}$ is internally stable, then there exists $\delta>0$ and $\rho>0$ such that if $0<\rho<\delta$ for all $X(0)\in \Se_I$ with $\|X(0)-x^*\|<\delta$, we have that the distribution of $X(t)$ converges as $t\to \infty$ to a unique invariant probability measure $\mu$, and $x^*\in supp(\mu)\subset cl(\Se_I)$ where $cl(\Se_I)$ is the closure of $\Se_I$.
\end{thm}

If $F_i(x,\bar e) = x_i f_i(x,\bar e)$ we define a \textit{saturated} equilibrium
$x^*\in \Se_I$ to be an equilibrium that satisfies $f_i(x^*)\leq 1$ for all $i\in I$. We say $x^*$ is an \textit{unsaturated} if $F_i(x^*)>1$ for some $i\in I$. 

\begin{rmk}
Note that since for an equilibrium $x^*$ we have
\[
x_i^* = x_i^* f_i(x^*, \bar e)
\]
this implies that if $x_i^*>0$ one needs $f_i(x^*, \bar e)=1$. Unsaturated equilibria can exist only for fixed points on the boundary $\Se_0$. This implies that if we use the result below for fixed points which are in $\Se_+$ we need not worry about saturated/unsaturated fixed points - they will automatically be saturated.
\end{rmk}

\begin{thm}\label{t:small_noise_A} (Theorems 3.3.2 and 3.3.1 from \cite{cuello2019persistence})
Suppose Assumption \ref{a:Cuello2} holds. There exists $\delta>0$ and a saturated $x^*\in cl(\Se_I)$ such that if $\BX(0)\in \Se_I$ and $0<\rho<\delta$ is sufficiently small then $\BX(t)\in B_\delta(x^*)$ for all $t$ large enough. Moreover, if there are no internally unstable saturated fixed points then for $\delta$ small enough with probability one the occupational measures of $\BX(t)$ converge weakly to an invariant probability measure associated (like in Theorem \ref{sfp_stoch2}) with an internally stable saturated fixed point $x^*$.

For the last part, assume we know that the occupational measures converge to the measure $\mu$ associated with the internally stable saturated fixed point $\NN$.  Let $D_x F$ be the Jacobian at $ \NN$ with respect to the population i.e. the $i-j$th entry of $D_x F$ is $\frac{\partial F_i}{\partial x_j}( \NN)$. Similarly $D_e F$ will be the matrix with entries $\diag(\frac{\partial F_i}{\partial e_i}(\NN))$.
Let $\NN^\rho$ be a vector with distribution $\mu$. We have
\[
\E N_i^{\rho} =  N_i + \sum_{l=1}^n\sum_{j_1, j_2=1}^n b_{il}\frac{\partial^2 F_l}{\partial x_{j_1}\partial x_{j_2}}A_{j_1,j_2} +  \sum_{l\in I}\sum_{j_1, j_2=1}^\ell b_{il}\frac{\partial^2 F_l}{\partial e_{j_1}\partial e_{j_2}} Cov(e_{j_1}, e_{j_2}) + O(\rho^3)
 \]
where $(b_{il})$ are the entries of the matrix $(I-D_xF( \NN))^{-1}$ and $(A_{ij})$ the entries of the matrix defined by
\[
vec(A) = (I-D_xF \otimes D_xF)^{-1} (D_eF \otimes D_eF) vec(Cov(r)).
\]

\end{thm}

Note that Theorem \ref{t:small_noise} follows as a corollary of Theorem \ref{t:small_noise_A} since there it is assumed there is only one stable equilibrium (and it is globally attracting) in $\Se_+$.

\bibliography{LV}

\end{document}